%% file: main.tex
\documentclass[a4paper,UKenglish,cleveref,autoref,final]{lipics-v2019}

\usepackage{xspace}
\usepackage{bm}
\usepackage{hyperref}
\usepackage{csquotes}
\usepackage{color}
\usepackage{microtype}
\usepackage{nicefrac}
\usepackage{mathtools}
\usepackage[inline]{enumitem}
\usepackage[sort&compress,numbers]{natbib}

\usepackage{tikz}
\usetikzlibrary{fit,patterns,decorations.pathreplacing}

\input{macros/complexity.tex}

\input{macros/todo.tex}

\input{macros/mics.tex}
\input{macros/nfold.tex}
\input{macros/appHandler.tex}

\theoremstyle{theorem}
\makeatletter
\newtheorem*{rep@theorem}{\rep@title}
\newcommand{\newreptheorem}[2]{%
\newenvironment{rep#1}[1]{%
 \def\rep@title{#2 \ref{##1}}%
 \begin{rep@theorem}}%
 {\end{rep@theorem}}}
\makeatother

\newreptheorem{theorem}{Theorem}
\newreptheorem{lemma}{Lemma}

\title{Scheduling Kernels via Configuration LP} 

\author{Dušan Knop}{Department of Theoretical Computer Science, Faculty of Information Technology, Czech Technical University in Prague, Czech Republic}{dusan.knop@fit.cvut.cz}{https://orcid.org/0000-0002-1825-0097}{}

\author{Martin Koutecký}{Computer Science Institute of Charles University, Charles University, Czech Republic}{koutecky@iuuk.mff.cuni.cz}{[orcid]}{}

\authorrunning{D.\ Knop and M.\ Koutecký}

\Copyright{Dušan Knop and Martin Koutecký}

\ccsdesc[500]{Theory of computation~Fixed parameter tractability}
\ccsdesc[500]{Theory of computation~Scheduling algorithms}
\ccsdesc[500]{Theory of computation~Integer Programming}

\keywords{Scheduling, Kernelization}

\category{}

\relatedversion{}

\supplement{}

\acknowledgements{
Authors wish to thank the Lorentz Center for making it possible to organize the workshop Scheduling Meets Fixed-Parameter Tractability, which output resulted in this paper.
The contribution of the Lorentz Center in stimulating suggestions, giving feedback and taking care of all practicalities, helped us to focus on our research and to organize a meeting of high scientific quality.
Furthermore, the authors thank Stefan Kratsch for bringing Proposition~\ref{prop:compressionKernelEquivalenceForNP} to their attention.
}

\nolinenumbers 

\hideLIPIcs  

\EventEditors{John Q. Open and Joan R. Access}
\EventNoEds{2}
\EventLongTitle{42nd Conference on Very Important Topics (CVIT 2016)}
\EventShortTitle{CVIT 2016}
\EventAcronym{CVIT}
\EventYear{2016}
\EventDate{December 24--27, 2016}
\EventLocation{Little Whinging, United Kingdom}
\EventLogo{}
\SeriesVolume{42}
\ArticleNo{23}

\begin{document}

\maketitle

\begin{abstract}
  Makespan minimization (on parallel identical or unrelated machines) is arguably the most natural and studied scheduling problem.
  A common approach in practical algorithm design is to reduce the size of a given instance by a fast preprocessing step while being able to recover key information even after this reduction.
  This notion is formally studied as kernelization (or simply, kernel) -- a polynomial time procedure which yields an equivalent instance whose size is bounded in terms of some given parameter.
  It follows from known results that makespan minimization parameterized by the longest job processing time $p_{\max}$ has a kernelization yielding a reduced instance whose size is exponential in $p_{\max}$.
  Can this be reduced to polynomial in $p_{\max}$?

  We answer this affirmatively not only for makespan minimization, but also for the (more complicated) objective of minimizing the weighted sum of completion times, also in the setting of unrelated machines when the number of machine kinds is a parameter.

  Our algorithm first solves the Configuration LP and based on its solution constructs a solution of an intermediate problem, called huge $N$-fold integer programming.
  This solution is further reduced in size by a series of steps, until its encoding length is polynomial in the parameters.
  Then, we show that huge $N$-fold IP is in \NP, which implies that there is a polynomial reduction back to our scheduling problem, yielding a kernel.

  Our technique is highly novel in the context of kernelization, and our structural theorem about the Configuration LP is of independent interest.
  Moreover, we show a polynomial kernel for huge $N$-fold IP conditional on whether the so-called separation subproblem can be solved in polynomial time.
  Considering that integer programming does not admit polynomial kernels except for quite restricted cases, our ``conditional kernel'' provides new insight.
\end{abstract}

\listoftodos

\clearpage

\section{Introduction}
Kernelization, data reduction, or preprocessing: all of these refer to the goal of simplifying and reducing (the size of) the input in order to speed up computation of challenging tasks.
Many heuristic techniques are applied in practice, however, we seek a theoretical understanding in the form of procedure with guaranteed bounds on the sizes of the reduced data.
We use the notion of kernelization from parameterized complexity (cf.~\cite{Niedermeier06,CyganFKLMPPS15}), where along with an input instance $I$ we get a positive integer $k$ expressing the \emph{parameter value}, which may be the size of the sought solution or some structural limitation of the input.
A kernel is an algorithm running in time $\poly(|I|)$ which returns a reduced instance $I'$ of the same problem of size bounded in terms of $k$; we sometimes also refer to $I'$ as the kernel.

It is well known~\cite{CaiCDF97} that a problem admits a kernel if and only if it has an algorithm running in time $f(k) \poly(|I|)$ for some computable function $f$ (i.e., if it is fixed-parameter tractable, or \FPT, parameterized by $k$).
The ``catch'' is that this kernel may be very large (exponential or worse) in terms of $k$, while for many problems, kernels of size polynomial in $k$ are known.
This raises a fundamental question for any \FPT problem: does it have a polynomial kernel?
Answering this question typically provides deep insights into a problem and the structure of its solutions.

Parameterized complexity has historically focused primarily on graph problems, but it has been increasingly branching out into other areas.
Kernelization, as arguably the most important subfield of parameterized complexity (cf. a recent monograph~\cite{FominLSZ19}), follows suit.
Scheduling is a fundamental area in combinatorial optimization, with results from parameterized complexity going back to 1995~\cite{BodlaenderF95}.
Arguably the most central problem in scheduling is makespan minimization on identical machines, denoted as $P||C_{\max}$, which we shall define soon.
It took until the seminal paper of Goemans and Rothvoss~\cite{GoemansRothvoss2014} to get an \FPT algorithm for $P||C_{\max}$ parameterized by the number of job types (hence also by the largest job).
Yet, the existence of a polynomial kernel for $P||C_{\max}$ remained open, despite being raised by Mnich and Wiese~\cite{MnichW15} and reiterated by van Bevern\footnote{The question was asked at the workshop ``Scheduling \& FPT'' at the Lorentz Center, Leiden, in February 2019, as a part of the opening talk for the open problem session.}.
Here, we give an affirmative answer for this problem:
\begin{corollary}\label{cor:pcmaxkernel}
	There is a polynomial kernel for $P||C_{\max}$ when parameterized by the longest processing time $p_{\max}$.
\end{corollary}
Let us now introduce and define the scheduling problems $P || C_{\max}$ and $P || \sum w_j C_j$.
There are $n$ \emph{jobs} and $m$ identical \emph{machines}, and the goal is to find a schedule minimizing an objective.
For each job $j \in [n]$, a \emph{processing time} $p_j \in \N$ is given and a \emph{weight} $w_j$ are given; in the case of $P || C_{\max}$ the weights play no role and can be assumed to be all zero.
A \emph{schedule} is a mapping which to each job $j \in [n]$ assigns some machine $i \in [m]$ and a closed interval of length $p_j$, such that the intervals assigned to each machine do not overlap except for their endpoints.
For each job $j \in [n]$, denote by $C_j$ its \emph{completion time}, which is the time when it finishes, i.e., the right end point of the interval assigned to $j$ in the schedule.
In the \emph{makespan minimization} ($C_{\max}$) problem, the goal is to find a schedule minimizing the time when the last job finishes $C_{\max} = \max_{j \in [n]} C_j$, called the \emph{makespan}.
In the \emph{minimization of sum of weighted completion times} ($\sum w_j C_j$), the goal is to minimize $\sum w_j C_j$.
(In the rest of the paper we formally deal with \emph{decision} versions of these problems, where the task is to decide whether there exists a schedule with objective value at most $k$.
This is a necessary approach when speaking of kernels and complexity classes like \NP and \FPT.)

In fact, our techniques imply results stronger in three ways, where we handle: \begin{enumerate\sv{*}}[label=\bfseries(\arabic*)] \item the much more complicated $\sum w_j C_j$ objective function involving possibly large job weights, \item the unrelated machines setting (denoted $R || C_{\max}$ and $R || \sum w_j C_j$), and \item allowing the number of jobs and machines to be very large, known as the \emph{high-multiplicity} setting.\end{enumerate\sv{*}}
For this, we need further notation to allow for different kinds of machines.
For each machine $i \in [m]$ and job $j \in [n]$, a \emph{processing time} $p_j^i \in \N$ is given.
For a given scheduling instance, say that two jobs $j, j' \in [n]$ are of the same \emph{type} if $p_j^i = p_{j'}^i$ for all $i \in [m]$ and $w_j = w_{j'}$, and say that two machines $i, i' \in [m]$ are of the same \emph{kind} if $p_j^i = p_j^{i'}$ for all jobs $j \in [n]$.
We denote by $\tau \in \N$ and $\kappa \in \N$ the number of job types and machine kinds, respectively, call this type of encoding the \emph{high-multiplicity} encoding, and denote the corresponding problems $R|HM|C_{\max}$ and $R|HM|\sum w_j C_j$.

Our approach is indirect: taking an instance $I$ of scheduling, we produce a small equivalent instance $I'$ of a the so-called \emph{huge $N$-fold integer programming} problem with a quadratic objective function (see more details below).
This is known as \emph{compression}, i.e., a polynomial time algorithm producing from $I$ a small equivalent instance of a different problem:
\newcommand{\RHighMultiplicityPolykernelStatement}{
	The problems $R|HM|C_{\max}$ and $R|HM|\sum w_j C_j$ parameterized by the number of job types $\tau$, the longest processing time $p_{\max}$, and the number of machine kinds $\kappa$ admit a polynomial compression to quadratic huge $N$-fold IP parameterized by the number of block types $\bar{\tau}$, the block dimension $t$, and the largest coefficient $\|E\|_\infty$.
}
\begin{theorem}
	\label{thm:RHighMultiplicityPolykernel}
	\RHighMultiplicityPolykernelStatement
\end{theorem}
If we can then find a polynomial reduction from quadratic huge $N$-fold IP to our scheduling problems, we are finished.
For this, it suffices to  show \NP membership, as we do in Lemma~\ref{lem:hugenfoldnp}.

\subparagraph{Configuration LP.} Besides giving polynomial kernels for some of the most fundamental scheduling problems, we wish to highlight the technique behind this result, because it is quite unlike most techniques used in kernelization and is of independent interest.
Our algorithm essentially works by solving the natural Configuration LP of $P || C_{\max}$ (and other problems), which can be done in polynomial time when $p_{\max}$ is polynomially bounded, and then using powerful structural insights to reduce the scheduling instance based on the Configuration LP solution.
The Configuration LP is a fundamental tool in combinatorial optimization which goes back to the work of Gilmore and Gomory in 1961~\cite{GilmoreGomory1961}.
It is known to provide high-quality results in practice\mkcom{cite}, in fact, the ``modified integer round-up property (MIRUP)'' conjecture states that the natural Bin Packing Configuration LP always attains a value which is at most one larger than the integer optimum~\cite{ScheitnerT:1995}.
The famous approximation algorithm of Karmarkar and Karp~\cite{KarmarkarKarp82} for Bin Packing is based on rounding the Configuration LP, and many other results in approximation use the Configuration LP for their respective problems as the starting point.

In spite of this centrality and vast importance of the Configuration LP, there are only few structural results providing deeper insight.
Perhaps the most notable is the work of Goemans and Rothvoss~\cite{GoemansRothvoss2014} and later Jansen and Klein~\cite{JansenK15} who show that there is a certain set of ``fundamental configurations'' such that in any integer optimum, all but few machines (bins, etc.) will use these fundamental configurations.
Our result is based around a theorem which shows a similar yet orthogonal result and can be informally stated as follows:
\begin{theorem}\label{thm:hugeNFoldKernel}
There is an optimum of the Configuration IP where all but few configuration are those discovered by the Configuration LP, and the remaining configurations are not far from those discovered by the Configuration LP.
\end{theorem}
We note that our result, unlike the ones mentioned above~\cite{GoemansRothvoss2014,JansenK15}, also applies to arbitrary separable convex functions.
This has a fundamental reason: the idea behind both previous results is to shift weight from the inside of a polytope to its vertices without affecting the objective value, which only works for linear objectives.

\subparagraph{Huge $N$-fold IP.} Finally, we highlight that the engine behind our kernels, a conditional kernel for the so-called quadratic huge $N$-fold IP, is of independent interest.
Integer programming is a central problem in combinatorial optimization.
Its parameterized complexity has been recently intensely studied~\cite{EisenbrandEtAl2019,KnopKM19,KouteckyLO18,EisenbrandHK18,ChanCKKP19}.
However, it turns out that integer programs cannot be kernelized in all but the most restricted cases~\cite{Kratsch13,Kratsch16,JansenK15}.
We give a positive result about a class of block-structured succinctly encoded IPs with a quadratic objective function, so-called quadratic huge $N$-fold IPs, which was used to obtain many interesting \FPT results~\cite{KnopK:2017,KnopKM19,Bredereck0KN19,GavenciakKK18,KnopKM17,BulteauHKLV20}.
However, our result is conditional on having a polynomial algorithm for the so-called \emph{separation subproblem} of the Configuration LP of the quadratic huge $N$-fold IP, so there is a price to pay for the generality of this fragment of IP.
The separation subproblem is to optimize a certain objective function (which varies) over the set of configurations.
In the cases considered here, we show that this corresponds to (somewhat involved) variations of the knapsack problem with polynomially bounded numbers; in other problems expressible as $n$-fold IP, the separation subproblem corresponds to a known hard problem.
Informally, our result reads as follows:
\newcommand{\hugeNFoldSeparationKernelStatement}{
	If~\eqref{eq:separation} is solvable in polynomial time, then quadratic huge $N$-fold IP admits a polynomial kernel when parameterized by $(\bar{\tau}, t, \|E\|_\infty)$.
}
\begin{theorem}\label{thm:redhugeip}
If the separation subproblem can be solved in polynomial time, then quadratic huge $N$-fold IP has a polynomial kernel parameterized by the block dimensions, the number of block types, and the largest coefficient.
\end{theorem}
One aspect of the algorithm above is reducing the quadratic objective function.
The standard approach, also used in kernelization of weighted problems~\cite{EtscheidKMR17,ChaplickFGK019,BentertvBFNN19,GoebbelsGRY17,BevernFT20,BevernFT19,GurskiRR19} is to use a theorem of Frank and Tardos~\cite{FT} which ``kernelizes'' a linear objective function if the dimension is a parameter.
However, we deal with \begin{enumerate\sv{*}}[label=\bfseries(\arabic*)] \item a quadratic convex (non-linear) function, \item over a space of large dimension. \end{enumerate\sv{*}}
We are able to overcome these obstacles by a series of steps which first ``linearize'' the objective, then ``aggregate'' variables of the same type, hence shrinking the dimension, then reduce the objective using the algorithm of Frank and Tardos, and then we carefully reverse this process (cf.~Lemma~\ref{lem:objreduction}).
This result has applications beyond this work: for example, the currently fastest strongly \FPT algorithm for $R||\sum w_j C_j$ (i.e., an algorithm whose number of arithmetic operations does not depend on the weights $w_j$) has dependence of $m^2 \poly\log(m)$ on the number of machines $m$; applying our new result instead of~\cite[Corollary 69]{EisenbrandEtAl2019} reduces this dependence to $m \poly\log(m)$.

\subparagraph*{Other Applications}
Theorem~\ref{thm:redhugeip} can be used to obtain kernels for other problems which can be modeled as huge $N$-fold IP.
First, we may also optimize the $\ell_p$ norms of times when each machine finishes, a problem known as $R|HM|\ell_p$.
Our results (Corollary~\ref{cor:conflp_g1_tw}) show that also in this setting the separation problem can be solved quickly.
Second, the $P||C_{\max}$ problem is identical to Bin Packing (in their decision form), so our kernel also gives a kernel for Bin Packing parameterized by the largest item size.
Moreover, also the Bin Packing with Cardinality Constraints problem has a huge $N$-fold IP model~\cite[Lemma 54]{KnopKLMO:2019} for which Corollary~\ref{cor:conflp_g1_tw} indicates that the separation subproblem can be solved quickly.
Third, Knop et al.~\cite{KnopKLMO:2019} give a huge $N$-fold IP model for the \textsc{Surfing}\mkcom{maybe make all problems textsc?} problem, in which many ``surfers'' make demands on few different ``services'' provided by few ``servers'', where each surfer may have different costs of getting a service from a server; one may think of internet streaming with different content types, providers, and pricing schemes for different customer types.
The separation problem there is polynomially solvable for an interesting reason: its constraint matrix is totally unimodular because it is the incidence matrix of the complete bipartite graph.
Thus, Theorem~\ref{thm:redhugeip} gives polynomial kernels for all of the problems above with the given parameters.

\subparagraph*{Related Work---Scheduling.}
Let us finally review related results in the intersection of parameterized complexity and scheduling.
First, up to our knowledge, to study scheduling problems from the perspective of multivariate complexity were Bodlaender and Fellows~\cite{BodlaenderF95}.
Fellows and McCartin~\cite{FellowsM03} study study scheduling on single machine of unit length jobs with (many) different release times and due dates.
Single machine scheduling where two agents compete to schedule their private jobs is investigated by Hermelin et al.~\cite{HermelinKSTW15}.
There are few other result~\cite{BevernMNW15,JansenMS17,HermelinST19,HermelinPST19} focused on identifying tractable scenarios for various scheduling paradigms (such as flow-shop scheduling or {e.g.}\ structural limitations of the job--machine assignment).

\section{Preliminaries}
We consider zero to be a natural number, i.e., $0 \in \N$.
We write vectors in boldface (e.g., $\vex, \vey$) and their entries in normal font (e.g., the $i$-th entry of a vector~$\vex$ is~$x_i$).
For positive integers $m \leq n$ we set $[m,n] \df \{m,\ldots, n\}$ and $[n] \df [1,n]$, and we extend this notation for vectors: for $\vel, \veu \in \Z^n$ with $\vel \leq \veu$, $[\vel, \veu] \df \{\vex \in \Z^n \mid \vel \leq \vex \leq \veu\}$ (where we compare component-wise).
For two vectors $\vex, \vey \in \R^n$, $\vez = \max\{\vex, \vey\}$ is defined coordinate-wise, i.e., $z_i = \max{x_i, y_i}$ for all $i \in [n]$, and similarly for $\min\{\vex, \vey\}$.

If~$A$ is a matrix, $A_{i,j}$ denotes the $j$-th coordinate of the $i$-th row, $A_{i, \bullet}$ denotes the $i$-th row and $A_{\bullet, j}$ denotes the $j$-th column.
We use $\log \df \log_2$, i.e., all our logarithms are base~$2$.
For an integer $a \in \Z$, we denote by $\la a \ra \df 1 + \ceil{\log (|a| + 1)}$ the binary encoding length of $a$; we extend this notation to vectors, matrices, and tuples of these objects.
For example, $\la A, \veb \ra = \la A \ra + \la \veb \ra$, and $\la A \ra = \sum_{i,j} \la A_{i,j} \ra$.
For a function $f\colon \Z^n \to \Z$ and two vectors $\vel, \veu \in \Z^n$, we define $f_{\max}^{[\vel, \veu]} \df \max_{\vex \in [\vel, \veu]} |f(\vex)|$; if $[\vel, \veu]$ is clear from the context we omit it and write just $f_{\max}$.

\toappendix{
  \lv{\subsection{Kernel and Compression}}\sv{\section{Kernel and Compression}}
  Let $(Q, \kappa)$ be a parameterized problem.
  We say that $(Q, \kappa)$ is \emph{fized-parameter tractable} (or in \textsf{FPT} for short) if there exists an algorithm that given an instance $(x, k)$ decides whether $(x,k) \in (Q, \kappa)$ in $f(k) \cdot \poly(|x|)$ time, where $f \colon \mathbb{N} \to \mathbb{N}$ is a computable function.
  A \emph{kernel} for $(Q, \kappa)$ is a polynomial time algorithm (that is, an algorithm that stops in $\poly(|x|)$ time) that given an instance $(x, k)$ returns an equivalent instance $(x', k')$ (that is, $(x', k') \in (Q,\kappa)$ if and only if $(x, k) \in (Q, \kappa)$) for which both $|x'|$ and $k'$ are upper-bounded by $g(k)$ for some computable function $g \colon \mathbb{N} \to \mathbb{N}$.
  It is well-known that a parameterized problem is in \textsf{FPT} if and only if there is a kernel for it.
  Of course, the smaller the size of the instance returned by the kernelization algorithm the better; in particular, we are interested in deciding whether $g$ can be a polynomial in $k$ and if this is the case, we say there is a polynomial kernel for $(Q, \kappa)$.
  A \emph{compression} is a similar notion to kernel, that is, it is a polynomial time algorithm that given $(x,k)$ returns an instance $(y, \ell)$ with $|y|, \ell \le g(k)$, however, this time we allow $(y, \ell)$ to be an instance of a different parameterized problem $(R, \lambda)$ and we require $(y, \ell) \in (R, \lambda)$ if and only if $(x, k) \in (Q, \kappa)$.
  A problem $(Q, \kappa)$ admits a polynomial compression if the function $g$ is a polynomial and we say that the problem $(Q, \kappa)$ admits a polynomial compression into the problem~$(R, \lambda)$.
}

\begin{proposition}[{\cite[Theorem~1.6]{FominLSZ19}}]\label{prop:compressionKernelEquivalenceForNP}
  Let $(Q, \kappa), (R, \lambda)$ be parameterized problems such that $Q$ is \textsf{NP}-hard and $R$ is in~\textsf{NP}.
  If $(Q, \kappa)$ admits a polynomial compression into $(R, \lambda)$, then it admits a polynomial kernel.
\end{proposition}
The above observation is useful when dealing with \textsf{NP}-hard problems.
The proof simply follows by pipelining the assumed polynomial compression with a polynomial time (Karp) reduction from $R$ to $Q$.

\subsection{Scheduling Notation}
Overloading the convention slightly, for each $i \in [\kappa]$ and $j \in [\tau]$, denote by $p_j^i$ the processing time of a job of type $j$ on a machine of kind $i$, by $w_j$ the weight of a job of type $j$, by $n_j$ the number of jobs of type $j$, by $m_i$ the number of machines of kind $i$, and denote $\ven = \left(n_1, \dots, n_\tau\right)$, $\vem = \left(m_1, \dots, m_\kappa\right)$, $\vep = \left(p^1_1, \dots, p^1_\tau, p^2_1, \dots, p_\tau^\kappa\right)$, $\vew \df \left(w_1, \dots, w_\tau\right)$, $p_{\max} \df \|\vep\|_\infty$, and $w_{\max} \df \|\vew\|_\infty$.
We denote the high multiplicity versions of the previously defined problems $R | HM | C_{\max}$ and $R | HM | \sum w_j C_j$.

For an instance $I$ of $R||C_{\max}$ or $R||\sum w_j C_j$, we define its size as $\la I \ra \df \sum_{i = 1}^\kappa \sum_{j = 1}^\tau \la p_j^i, w_j \ra$, whereas for an instance $I$ of $P|HM|C_{\max}$ or $P|HM|\sum w_j C_j$ we define its size as $\la I \ra = \la \ven, \vem, \vep, \vew \ra$.
Note that the difference in encoding actually leads to different problems: for example, an instance of $R | HM | C_{\max}$ with $2^k$ jobs with maximum processing time $p_{\max}$ can be encoded with $\Oh(k \tau \kappa \log p_{\max})$ bits while an equivalent instance of $R || C_{\max}$ needs $\Omega(2^k \log p_{\max})$ bits, which is exponentially more if $\tau, \kappa \in k^{\Oh(1)}$.
The membership of high-multiplicity scheduling problems in \NP was open for some time, because it is not obvious whether a compactly encoded instance also has an optimal solution with a compact encoding.
This question was considered by Eisenbrand and Shmonin, and we shall use their result.
For a set $X \subseteq \Z^d$ define the \emph{integer cone of $X$}, denoted $\intcone(X)$, to be the set
\(
\intcone(X) \df \left\{ \sum_{\vex \in X} \lambda_{\vex} \vex \mid \lambda \in \N^{X} \right\} \,.
\)
\begin{proposition}[{Eisenbrand and Shmonin~\cite[Theorem 2]{EisenbrandS:2006}}] \label{prop:es}
	Let $X \subseteq \Z^d$ be a finite set of integer vectors and let $\veb \in \intcone(X)$.
	Then there exists a subset $\tilde{X} \subseteq X$ such that $\veb \in \intcone(\tilde{X})$ and the following holds for the cardinality of $\tilde{X}$:
	\begin{enumerate}
		\item if all vectors of $X$ are nonnegative, then $|\tilde{X}| \leq \la \veb \ra$,
		\item if $M = \max_{\vex \in X} \|\vex\|_\infty$, then $|\tilde{X}| \leq 2d(\log 4dM)$.
	\end{enumerate}
\end{proposition}
One can use Proposition~\ref{prop:es} to show that the decision versionf of $R |HM | C_{\max}$ and $R |HM | \sum w_j C_j$ have short certificates and thus belong to \NP.
We will later derive the same result as a corollary of the fact that both of these scheduling problems can be encoded as a certain form of integer programming, which we will show to have short certificates as well.

\subsection{Conformal Order and Graver Basis}
Let $\veg, \veh \in \mathbb{Z}^n$ be two vectors.
We say that \emph{$\veg$ is conformal to $\veh$} (we denote it $\veg \sqsubseteq \veh$) if both $g_i \cdot h_i \ge 0$ and $|g_i| \le |h_i|$ for all $i \in [n]$.
In other words, $\veg \sqsubseteq \veh$ if they are in the same orthant (the first condition holds) and $\veg$ is component-wise smaller than~$\veh$.
For a matrix~$A$ we define its \emph{Graver basis}~$\G(A)$ to be the set of all $\sqsubseteq$-minimal vectors in $\Ker(A) \setminus \{ \vezero \}$.
We define \( g_\infty(A) = \max \left\{ \left\| \veg \right\|_\infty \mid \veg \in \G(A) \right\} \) and \( g_1(A) = \max \left\{ \left\| \veg \right\|_1 \mid \veg \in \G(A) \right\} \).

We say that two functions $f,g \colon \Z^d \to \Z$ are \emph{equivalent} on a polyhedron $P \subseteq \Z^d$ if $f(\vex) \leq f(\vey)$ if and only if $g(\vex) \leq g(\vey)$ for all $\vex, \vey \in P$.
Note that if~$f$ and~$g$ are equivalent on~$P$, then the set of minimizers of~$f(\vex)$ over~$P$ is the same as the set of minimizers of~$g(\vex)$ over~$P$.
\begin{proposition}[{Frank and Tardos~\cite{FT}}]\label{thm:FT}
	Given a rational vector $\vew \in \Q^{d}$ and an integer $M$, there is a polynomial algorithm which finds a $\tilde{\vew} \in \Z^d$ such that the linear functions $\vew \vex$ and $\tilde{\vew}\vex$ are equivalent on $[-M,M]^d$, and $\|\tilde{\vew}\|_\infty \leq 2^{\Oh(d^3)} M^{\Oh(d^2)}$.
\end{proposition}
The \emph{dual graph} $G_D(A) = (V,E)$ of a matrix $A \in \Z^{m \times n}$ has $V=[m]$ and $\{i,j\} \in E$ if rows $i$ and $j$ contain a non-zero at a common coordinate $k \in [n]$.
The dual treewidth $\tw_D(A)$ of $A$ is $\tw(G_D(A))$.
We do not define treewidth here, but we point out that $\tw(T) = 1$ for every tree~$T$.
\begin{proposition}[Eisenbrand et al.~{\cite[Theorem 98]{EisenbrandEtAl2019}}] \label{prop:fastip}
	An IP with a constraint matrix $A$ can be solved in time $(\|A\|_\infty g_1(A))^{\Oh(tw_D(A))} \poly(n, L)$, where $n$ is the dimension of the IP and $L$ is the length of the input.
\end{proposition}
\begin{proposition}[Eisenbrand et al.~{\cite[Lemma 25]{EisenbrandEtAl2019}}] \label{prop:basebound}
	For an integer matrix $A \in \Z^{m \times n}$, we have $g_1(A) \leq (2\|A\|_\infty m +1)^m$.
\end{proposition}

\lv{
Let us use Proposition~\ref{prop:es} to show that $R |HM | C_{\max}$ and $R |HM | \sum w_j C_j$ have short certificates.
Here and later we will use the notion of a configuration: a \emph{configuration} is a vector $\vecc \in \N^\tau$ encoding how many jobs of which type are assigned to some machine.
	\begin{lemma}\label{lem:npmembership}
		(The decision versions of) $R |HM | C_{\max}$ and $R |HM | \sum w_j C_j$ belong to \NP.
	\end{lemma}
\begin{proof}
To show membership in \NP, we have to prove the existence of short certificates.
More precisely, for a high-multiplicity scheduling instance $I$ with a parameter $OPT$, we have to show that if $I$ has an optimum of at most $OPT$, then there exists a certificate of this fact of length $\poly(\la I \ra)$.
In both cases ($R | HM | C_{\max}$ and $R | HM | \sum w_j C_j$) the certificate will be a collection of configurations together with their multiplicities.
However, to use Proposition~\ref{prop:es} we will need to introduce a more complicated notion of an extended configuration.
\medskip\noindent$R | HM | C_{\max}$.
Let an instance $I$ of $R | HM | C_{\max}$ together with the value~$OPT$ be given.
For each machine kind $i \in [\kappa]$, define $X^i_{C_{\max}} \df \left\{\vecc \in \N^\tau \mid \vep^i \vecc \leq OPT\right\} \times \{0\}^{i-1} \times \{1\} \times \{0\}^{\kappa - i}$, and define the set of its \emph{extended configurations} of~$I$ to be $X_{C_{\max}} \df \bigcup_{i=1}^\kappa X^i_{C_{\max}}$.
The interpretation is that in any $\vex \in X_{C_{\max}}$ the first $\tau$ coordinates encode a configuration (i.e., an assignment of jobs to a machine) and the remaining $\kappa$ coordinates encode the kind of a machine for which this configuration can be processed in time at most $OPT$.
Then any decomposition of the vector $(\ven, \vem) = \sum_{\vex \in X_{C_{\max}}} \lambda_{\vex} \vex$ with $\velambda \in \N^{X_{C_{\max}}}$ corresponds to a solution of $I$ where the last job finishes in time at most $OPT$.
Finally, since all vectors in~$X_{C_{\max}}$ are nonnegative, Proposition~\ref{prop:es} (Part 1) applied to $\intcone(X_{C_{\max}})$ says that if such a decomposition exists (i.e., if~$I$ is a \textsc{Yes} instance), then there exists one with $|\suppo(\velambda)| \leq \la (\ven, \vem) \ra \leq \la I \ra$ and we are done.

\medskip\noindent$R | HM | \sum w_j C_j$.
Let $I$ be an instance of $R | HM | \sum w_j C_j$ together with the value~$OPT$.
It is well known~\cite{Smith1956} that on a single machine a schedule minimizing $\sum w_j C_j$ is one which schedules jobs according to their \emph{Smith ratios} $w_j / p_j^i$ non-increasingly.
For each machine kind $i \in [\kappa]$, we define $f^i \colon \N^\tau \to \N$ to be the value of $\sum w_j C_j$ for the aforementioned scheduling of the instance $\vecc$ on a single machine of kind $i$.
Define $X^i_{\sum w_j C_j} \df \{\left(\vecc, F\right) \in \N^{\tau+1} \mid \vecc \leq \ven, f^i(\vecc) \leq F \leq OPT\} \times \{0\}^{i-1} \times \{1\} \times \{0\}^{\kappa - i}$, and define the set of \emph{extended configurations} to be $X_{\sum w_j C_j} \df \bigcup_{i=1}^\kappa X^i_{\sum w_j C_j}$.
The difference, as compared with $C_{\max}$, is that $OPT$ does not define $X_{\sum w_j C_j}$ (we only use it to ensure finiteness) but we have an additional coordinate which expresses (an upper bound on) the contribution of each configuration (machine) to the objective.
Hence, any decomposition of the vector $(\ven, OPT, \vem) = \sum_{\vex \in X_{\sum w_j C_j}} \lambda_{\vex} \vex$ with $\velambda \in \N^{X_{\sum w_j C_j}}$ corresponds to a solution of $I$ of value at most $OPT$.
Proposition~\ref{prop:es} says if any decomposition exists (i.e., if $I$ is a \textsc{Yes} instance), then there exists one where $|\suppo(\velambda)| \leq \la (\ven, OPT, \vem) \ra$.
Because $f$ is a quadratic function with coefficients bounded by $p_{\max}$ and $w_{\max}$~\cite{KnopK:2017} we have $\la OPT \ra \leq \la I \ra^2$ and hence there exists a certificate of length $\poly(\la I \ra)$ and $R | HM | \sum w_j C_j$ is in \NP.
\end{proof}
}

\subsection{$N$-fold Integer Programming}
The \textsc{Integer Programming} problem is to solve:
\begin{equation}
\min f(\vex):\, A\vex = \veb, \, \vel \leq \vex \leq \veu, \, \vex \in \Z^n,\label{IP} \tag{IP}
\end{equation}
where $f\colon \R^n \to \R$, $A \in \Z^{m \times n}$, $\veb \in \Z^m$, and $\vel, \veu \in (\Z \cup \{\pm \infty\})^n$.

A \emph{generalized $N$-fold IP matrix} is defined as
\begin{align}
E^{(N)} = \left(
\begin{array}{cccc}
E^1_1    & E^2_1    & \cdots & E^N_1    \\
E^1_2    & 0      & \cdots & 0      \\
0    & E^2_2    & \cdots & 0      \\
\vdots & \vdots & \ddots & \vdots \\
0    & 0      & \cdots & E^N_2    \\
\end{array}
\right) \enspace .\label{nfold}
\end{align}
Here, $r,s,t,N \in \N$, $E^{(N)}$ is an $(r+Ns)\times Nt$-matrix, and $E^i_1 \in \Z^{r \times t}$ and $E^i_2 \in \Z^{s \times t}$ for all $i \in [N]$, are integer matrices.
Problem~\eqref{IP} with $A=E^{(N)}$ is known as \emph{generalized $N$\hy fold integer programming} (generalized $N$-fold IP).
``Regular'' $N$-fold IP is the problem where $E_1^i = E_1^j$ and $E_2^i = E_2^j$ for all $i, j \in [N]$.
Recent work indicates that the majority of techniques applicable to ``regular'' $N$-fold IP also applies to generalized $N$-fold IP~\cite{EisenbrandEtAl2019}.

The structure of $E^{(N)}$ allows us to divide any $Nt$-dimensional object, such as the variables of $\vex$, bounds $\vel, \veu$, or the objective $f$, into $N$ \textit{bricks} of size $t$, e.g. $\vex= \left(\vex^1, \dots, \vex^N\right)$.
We use subscripts to index within a brick and superscripts to denote the index of the brick, i.e.,~$x_j^i$ is the $j$-th variable of the $i$-th brick with $j \in [t]$ and $i \in [N]$.
We call a brick \emph{integral} if all of its coordinates are integral, and \emph{fractional} otherwise.

\subparagraph{\texorpdfstring{Huge $N$-fold}{Huge N-fold} IP.}
The \emph{huge $N$-fold IP} problem is an extension of generalized $N$-fold IP to the high-multiplicity scenario, where blocks come in types and are encoded succinctly by type multiplicities.
This means there could be an \emph{exponential} number of bricks in an instance with a polynomial encoding size.
The input to the huge $N$-fold IP problem with $\bar{\tau}$ \emph{types of blocks} is defined by matrices $E^i_1 \in \Z^{r \times t}$ and $E^i_2 \in \Z^{s \times t}$, $i \in [\bar{\tau}]$, vectors $\vel^1, \dots, \vel^{\bar{\tau}}$, $\veu^1, \dots, \veu^{\bar{\tau}} \in \Z^t$, $\veb^0 \in \Z^r$, $\veb^1, \dots, \veb^{\bar{\tau}} \in \Z^s$, functions $f^1, \dots, f^{\bar{\tau}} \colon \R^{t} \to \R$ satisfying $\forall i \in [\bar{\tau}], \, \forall \vex \in \Z^t$ we have $f^i(\vex) \in \Z$ and given by evaluation oracles, and integers $\mu^1, \dots, \mu^{\bar{\tau}} \in \N$ such that $\sum_{i=1}^{\bar{\tau}} \mu^i = N$.
We say that a brick is of type $i$ if its lower and upper bounds are $\vel^i$ and $\veu^i$, its right hand side is $\veb^i$, its objective is $f^i$, and the matrices appearing at the corresponding coordinates are $E^i_1$ and $E^i_2$.
Denote by $T_i$ the indices of bricks of type $i$, and note $|T_i| = \mu_i$ and $|\bigcup_{i \in [\bar{\tau}]} T_i| = N$.
The task is to solve~\eqref{IP} with a matrix $E^{(N)}$ which has $\mu^i$ blocks of type $i$ for each $i$.
Knop et al.~\cite{KnopKLMO:2019} have shown a fast algorithm solving huge $n$-fold IP.
The main idea of their approach is to prove a powerful proximity theorem showing how one can drastically reduce the size of the input instance given that one can solve a corresponding configuration LP (which we shall formally define later).
We will build on this approach here.
When $f^i$ are restricted to be separable quadratic (and convex) for all $i \in [\bar{\tau}]$, we call the problem \emph{quadratic huge $N$-fold IP}.

\subsection{Configurations LP of \texorpdfstring{Huge $N$-fold IP}{Huge N-fold IP}} \label{ssec:conflp}
Having modeled our scheduling problems as huge $N$-fold IP instances, our next goal is to solve the Configuration LP, which we will now define.
Because the results we derive below apply to any quadratic huge $N$-fold IP, we state them generally (and not as claims about the specific instances which we shall apply them to).

Let a huge $N$-fold IP instance with $\bar{\tau}$ types be fixed.
Recall that $\mu^i$ denotes the number of blocks of type $i$, and let $\vemu = \left(\mu^1, \dots, \mu^{\bar{\tau}}\right)$.
We define for each $i \in [\bar{\tau}]$ the set of configurations of type $i$ as
\[
\CC^i = \left\{\vecc \in \Z^t \mid E^i_2 \vecc = \veb^i, \, \vel^i \leq \vecc \leq \veu^i \right\} \enspace .
\]
Here we are interested in four instances of convex programming (CP) and convex integer programming (IP) related to huge $N$-fold IP.
First, we have the \emph{Huge IP}
\begin{equation} \label{eq:hugenfold}
\min f(\vex): \, E^{(N)} \vex = \veb, \, \vel \leq \vex \leq \veu, \, \vex \in \Z^{Nt} \enspace . \tag{HugeIP}
\end{equation}
Then, there is the \emph{Configuration LP} of~\eqref{eq:hugenfold},
\begin{align}
\min \vev \vey & = \sum_{i=1}^{\bar{\tau}} \sum_{\vecc \in \CC^i} f^i(\vecc) \cdot y(i, \vecc) & \label{eq:conflp_start} \\
\sum_{i=1}^{\bar{\tau}} E^i_1 \sum_{\vecc \in \CC^i} \vecc y(i, \vecc) &= \veb^0 &\notag \\
\sum_{\vecc \in \CC^i} y(i, \vecc) &= \mu^i & \forall i \in [\bar{\tau}]\notag \\
\vey &\geq \mathbf{0} \enspace . & \label{eq:conflp_end}
\end{align}
Let $B$ be its constraint matrix and $\ved = (\veb^0, \vemu)$ be the right hand side and shorten \eqref{eq:conflp_start}-\eqref{eq:conflp_end} to
\begin{equation} \label{eq:conflp}
\min \vev \vey:\, B \vey = \ved, \, \vey \geq \mathbf{0} \enspace . \tag{ConfLP}
\end{equation}
Finally, by observing that $B\vey=\ved$ implies $y(i,\vecc) \leq \|\vemu\|_\infty$ for all $i \in [\bar{\tau}], \vecc \in \CC^i$, defining $C = \sum_{i \in [\bar{\tau}]} \left|\CC^i\right|$, leads to the \emph{Configuration ILP},
\begin{equation} \label{eq:confilp}
\min \vev \vey: \,B \vey = \ved, \, \mathbf{0} \leq \vey \leq (\|\vemu\|_\infty, \dots, \|\vemu\|_\infty),\, \vey \in \N^{C} \enspace . \tag{ConfILP}
\end{equation}

The classical way to solve~\eqref{eq:conflp} is by solving its dual using the ellipsoid method and then restricting~\eqref{eq:conflp} to the columns corresponding to the rows encountered while solving the dual, a technique known as column generation.
The Dual LP of~\eqref{eq:conflp} in variables $\bm{\alpha} \in \R^r$, $\bm{\beta} \in \R^{\bar{\tau}}$ is:
\begin{align}
\max  & & \veb^0 \bm{\alpha} + \sum_{i=1}^{\bar{\tau}}\mu^i \beta^i & \notag \\
\textrm{s.t.} & & (\bm{\alpha} E^i_1) \vecc - f^i(\vecc)              & \leq -\beta^i & \forall i \in [\bar{\tau}], \,\forall \vecc \in \CC^i \label{eq:dualcons}
\end{align}
To verify feasibility of $(\bm{\alpha}, \bm{\beta})$ for $i \in [\bar{\tau}]$, we need to maximize the left-hand side of~\eqref{eq:dualcons} over all $\vecc \in \CC^i$ and check if it is at most $-\beta^i$.
This corresponds to solving the following \emph{separation problem}: find integer variables $\vecc$ which for a given vector $(\bm{\alpha}, \bm{\beta})$ solve
\begin{equation}
\min f^i(\vecc)- (\bm{\alpha} E^i_1) \vecc \,:\, E^i_2 \vecc = \veb^i, \, \vel^i \leq \vecc \leq \veu^i, \vecc \in \Z^t \enspace . \tag{$\sep$-IP} \label{eq:separation}
\end{equation}
Denote by $\sep(\vel^i, \veu^i, f^i_{\max}, E_1^i, E_2^i)$ the time needed to solve~\eqref{eq:separation}.
\begin{lemma}[{\lv{Knop et al.~}\cite[Lemma 12]{KnopKLMO:2019}}]
	\label{lem:frac}
	An optimal solution $\vey^*$ of~\eqref{eq:conflp} with $|\suppo(\vey^*)| \leq r + \bar{\tau}$ can be found in
	\(
	(r t \bar{\tau} \la f_{\max}, \vel, \veu, \veb, \vemu \ra)^{\Oh(1)} \cdot \max_{i \in [\bar{\tau}]} \sep(\vel^i, \veu^i, f^i_{\max}, E_1^i, E_2^i)
	\) time.
\end{lemma}
Since~\eqref{eq:separation} is an IP, it can be solved using Proposition~\ref{prop:fastip} in time $g_1(E_2^i)^{\tw_D(E_2^i)} \cdot \poly(\vel^i, \veu^i, \veb^i, \|E_1^i\|_\infty f_{\max}, t, \bar{\tau})$.
Hence, together with Lemma~\ref{lem:frac}, we get the following corollary:
\begin{corollary}\label{cor:conflp_g1_tw}
	An optimal solution $\vey^*$ of~\eqref{eq:conflp} with $|\suppo(\vey^*)| \leq r + \bar{\tau}$ can be found in time
\(
(r t \bar{\tau} \la f_{\max}, \vel, \veu, \veb, \vemu \ra)^{\Oh(1)} \cdot \max_{i \in [\bar{\tau}]} g_1(E_2^i)^{\tw_D(E_2^i)}
\).
\end{corollary}
We later show how that for our formulations of $R|HM|C_{\max}$ and $R|HM|\sum w_j C_j$, indeed $g_1(E_2^i)$ is polynomial in $\tau, p_{\max}$, and $\tw_D(E_2^i) = 1$, hence the~\eqref{eq:conflp} optimum can be found in polynomial time.

\section{Compressing High Multiplicity Scheduling to Quadratic \texorpdfstring{$N$}{N}-fold IP}\label{sec:compressingSchedulingToNFold}
\sv{\toappendix{\section{Additional Material for Section~\ref{sec:compressingSchedulingToNFold}}}}

In this section we are going to prove Theorem~\ref{thm:RHighMultiplicityPolykernel}.
To that end, we use the following assumption (which mainly simplifies notation).

\begin{remark}
  From here on, we assume $\tau \geq \|\vep\|_\infty$, since both quantities are parameters.
\end{remark}

\begin{reptheorem}{thm:RHighMultiplicityPolykernel}[repeated]
  \RHighMultiplicityPolykernelStatement
\end{reptheorem}
Recall that in order to use Theorem~\ref{thm:RHighMultiplicityPolykernel} to provide \emph{kernels} for selected scheduling problems (which are \NP-hard) we want to utilize Proposition~\ref{prop:compressionKernelEquivalenceForNP}.
Thus, we have to show that the ``target problem'' quadratic huge $N$-fold IP is in \NP.
\lv{\begin{lemma}}\sv{\begin{lemma}[\appmark]}\label{lem:hugenfoldnp}
	The decision version of quadratic huge $N$-fold IP belongs to \NP.
\end{lemma}
\toappendix{
  \lv{\begin{proof}}\sv{\begin{proof}[Proof of Lemma~\ref{lem:hugenfoldnp}]}
  	We will use Proposition~\ref{prop:es} to show that there exists an optimum whose number of distinct configurations is polynomial in the input length.
  	Such a solution can then be encoded by giving those configurations together with their multiplicities, and constitutes a polynomial certificate.
  	Recall that~\eqref{eq:confilp} corresponding to the given instance of huge $N$-fold is
  	\[
  	\min \vev \vey: \,B \vey = \ved, \, \mathbf{0} \leq \vey \leq (\|\vemu\|_\infty, \dots, \|\vemu\|_\infty),\, \vey \in \N^{C} \enspace .
  	\]
  	Let $X$ be the set of columns of the matrix $B$ extended with an additional coordinate which is the coefficient of the objective function $\vev$ corresponding to the given column, that is, $\vev\veb$ for a column~$\veb$ (i.e., the objective value of configuration $\veb$).
  	Hence $X \subseteq \Z^{r+\bar{\tau} + 1}$ and $\|\vex\|_\infty \leq \|\vel, \veu, f_{\max}\|_\infty =: M$ for any~$\vex \in X$.
  	Applying Proposition~\ref{prop:es}, part~2, to $X$, yields that there exists an optimal solution $\vey$ of~\eqref{eq:confilp} with $\suppo(\vey) = \tilde{X}$ satisfying $|\tilde{X}| \leq 2(r+\bar{\tau} + 1) \log (4(r + \bar{\tau} + 1) M)$, hence polynomial in the input length of the original instance.
  \end{proof}

  \begin{remark}
    Clearly Lemma~\ref{lem:hugenfoldnp} holds for any huge $N$-fold IP whose objective is restricted by some, not necessarily quadratic, polynomial.
  \end{remark}
}

\subparagraph{Using Theorem~\ref{thm:RHighMultiplicityPolykernel}.}
Before we move to the proof of Theorem~\ref{thm:RHighMultiplicityPolykernel} we first derive two simple yet interesting corollaries.
\begin{corollary}\label{cor:rObjKernel}
  The problems $R||C_{\max}$ and $R||\sum w_jC_j$ admit polynomial kernelizations when parameterized by $\tau,\kappa, p_{\max}$.
\end{corollary}
\begin{proof}
  Let $\mathfrak{obj} \in \{C_{\max}, \sum w_j C_j\}$.
  We describe a polynomial compression from $R||\mathfrak{obj}$ to quadratic huge $N$-fold IP which, by Lemma~\ref{lem:hugenfoldnp}, yields the sought kernel, since $R||\mathfrak{obj}$ is \NP-hard and huge $N$-fold with a quadratic objective is in \NP.

  We first perform the high-multiplicity encoding of the given instance~$I$ of $R||\mathfrak{obj}$, thus obtaining an instance~$I_{HM}$ of $R|HM|\mathfrak{obj}$ with the input encoded as $(\ven, \vem, \vep, \vew)$.
  Now, we can apply Theorem~\ref{thm:RHighMultiplicityPolykernel} and obtain an instance~$\bar{I}_{\text{huge $N$-fold}}$ equivalent to~$I_{HM}$ with size bounded by a polynomial in $\kappa, \tau, p_{\max}$.
\end{proof}

\begin{proof}[Proof of Corollary~\ref{cor:pcmaxkernel}]
  This is now trivial, since it suffices to observe that $P||C_{\max}$ is a special case of $R||C_{\max}$, where there is only a single machine kind (i.e., $\kappa = 1$) and $\tau \le p_{\max}$ job types.
  Our claim then follows by Corollary~\ref{cor:rObjKernel} (combined with the fact that $P||C_{\max}$ is \NP-hard and $R||C_{\max}$ is in \NP).
\end{proof}

\subsection{Huge $n$-fold IP Models}\label{ssec:hugeNFoldIPModesHMScheduling}
Denote by $\ven_{\max}$ the $\tau$-dimensional vector whose all entries are $\|\ven\|_\infty$.
It was shown~\cite{KnopK:2017,KnopKLMO:2019} that $R|HM|C_{\max}$ is modeled as a feasibility instance of huge $n$-fold IP as follows.
Recall that we deal with the decision versions and that~$k$ is the upper bound on the value of the objective(s).
We set $\veb^0 = \ven$, the number of block types is $\bar{\tau} = \kappa$, $E_1^i := (I~\vezero) \in \Z^{\tau \times (\tau+1)}$, $E_2^i := (\vep^i~1)$, $\vel^i = \vezero, \veu^i = (\ven, \infty)$, $\veb^i = k$, for $i \in [\kappa]$, and the multiplicities of blocks are $\vemu = \vem$.
The meaning is that the first type of constraints expressed by the $E_1^i$ matrices ensures that every job is scheduled somewhere, and the second type of constraints expressed by the $E_2^i$ matrices ensures that every machine finishes in time $C_{\max}$.

In the model of $R|HM|\sum w_j C_j$, for each machine kind $i \in [\kappa]$, we define $\preceq^i$ to be the ordering of jobs by the $w_j/p^i_j$ ratio non-increasingly, and let $\vea = (a_1, a_2, \dots, a_\tau)$ be a reordering of $\vep^i$ according to $\preceq^i$.
We let
\[
G^i := \left(\begin{matrix}
a_1 & 0 & 0 & \dots & 0 \\
a_1 & a_2 & 0 & \dots & 0 \\
a_1 & a_2 & a_3 & \dots & 0 \\
\vdots & & & \ddots &  \\
a_1 & a_2 & a_3 &\dots & a_\tau
\end{matrix}\right), \qquad \qquad H := -I,
\]
with $I$ the $\tau \times \tau$ identity, and define $F^i := (G^i~H)$ in two steps.
Denote by $I^{\preceq^i}$ a matrix obtained from the $\tau \times \tau$ identity matrix by permuting its columns according to $\preceq^i$.
The model is then $\veb^0 = \ven$, the number of block types is again $\bar{\tau} = \kappa$,  for each $i \in [\kappa]$ we have $E_1^i = (I^{\preceq^i}~\vezero) \in \Z^{\tau \times 2\tau}$, $E_2^i = \bar{F}^i$, $\vel^i = \vezero$, $\veu^i = (\ven_{\max}, p_{\max}\tau \ven_{\max})$, $\veb^i = \vezero$, $f^i$ is a separable convex quadratic function (whose coefficients are related to the $w_j/p^i_j$ ratios), and again $\vemu = \vem$.
Intuitively, in each brick, the first $\tau$ variables represent numbers of jobs of each type on a given machine, and the second $\tau$ variables represent the amount of processing time spend by jobs of the first $j$ types with respect to the ordering $\preceq^i$.

\subsection{Solving The Separation Problem Quickly: $C_{\max}$}\label{ssec:solvingSeparationProblemQuicklyCmax}
The crucial aspect of complexity of~\eqref{eq:separation} is its constraint matrix $E_2^i$.
For $R|HM|C_{\max}$, this is just the vector $(\vep^i,1)$.
Clearly $tw_D((\vep^i,1)) = 1$ since $G_D((\vep^i,1))$ is a single vertex.
By Proposition~\ref{prop:basebound}, $g_1((\vep^i,1)) \leq 2\|\vep\|_\infty +1$.
Moreover, $f_{\max}$ depends polynomially on $\|\ven, \vem, \vep\|_\infty$.
Hence, Corollary~\ref{cor:conflp_g1_tw} states that~\eqref{eq:conflp} of the $R|HM|C_{\max}$ model can be solved in time
$(r t \bar{\tau} \la f_{\max}, \vel, \veu, \veb, \vemu \ra)^{\Oh(1)} \cdot \max_i g_1(E_2^i)^{\tw_D(E_2^i)} = \poly(p_{\max}, \tau, \log \|\ven, \vem, \vep\|_\infty)$, which is polynomial in the input.

\subsection{Solving The Separation Problem Quickly: $\sum w_j C_j$}\label{ssec:solvingSeparationProblemQuicklysumwjCj}
The situation is substantially more involved in the case of $R|HM|\sum w_j C_j$: in order to apply Corollary~\ref{cor:conflp_g1_tw}, we need to again bound $g_1(E_2^i)$ and $\tw_D(E_2^i)$, but the matrix $E_2^i$ is more involved now.
Let
\[
\bar{G}^i := \left(
\begin{matrix}
a_1 & 0 & \dots & 0 \\
0 & a_2 & \dots & 0 \\
\vdots &  & \ddots & \vdots \\
0 & 0 & \dots & a_\tau
\end{matrix}
\right), \qquad \qquad
\bar{H}:= \left(
\begin{matrix}
-1 & 0 & 0 & \dots & 0  \\
1 & -1 & 0 & \dots & 0  \\
0 & 1 & -1 & \dots & 0  \\
\vdots &  & \ddots & \ddots & \vdots \\
0 & 0 & \dots & 1 & -1
\end{matrix}
\right),
\]
and define $\bar{F}^i = (\bar{G}^i~\bar{H})$.
Now observe that $F^i$ and $\bar{F}^i$ are row-equivalent\footnote{Two matrices $A, A'$ are row-equivalent if one can be transformed into the other using elementary row operations.}.
This means that we can replace $F^i$ with $\bar{F}^i$ without changing the meaning of the constraints and without changing the feasible set.
But while $\tw_D(F_i) = \tau$ (because it is the clique $K_{\tau}$), we have
\lv{\begin{lemma}}\sv{\begin{lemma}[\appmark]}\label{lem:twD}
	For each $i \in [\kappa]$, $\tw_D(\bar{F}^i) = 1$.
\end{lemma}
\toappendix{
  \lv{\begin{proof}}\sv{\begin{proof}[Proof of Lemma~\ref{lem:twD}]}
  	Observe that the dual graph of~$\bar{H}$ is a path (on~$\tau$ vertices).
  	By the definition of the dual graph the~$\bar{G}^i$ part of~$\bar{F}^i$ does not add any edges to it.
  \end{proof}
}
Note that application of Proposition~\ref{prop:basebound} yields $g_1(E^i_2) \le \Oh(\tau^\tau)$.
This general upper bound, as we shall see, is not sufficient for our purposes, since we need $g_1(E_2^i) \leq \poly(\tau)$.
However, we can improve it significantly:
\begin{lemma}[Hill-cutting] \label{lem:hillcut}
	We have $g_\infty(F^i), g_\infty(\bar{F}^i) \leq \Oh(\tau^4)$ and $g_1(F^i), g_1(\bar{F}^i) \leq \Oh(\tau^5)$ for every $i \in [\kappa]$.
\end{lemma}
\begin{proof}
	Let $F = F^i$ for some $i \in [\kappa]$.
	Let $(\vex, \vez) \in \Z^{2\tau}$ be some vector satisfying $F \cdot (\vex, \vez) = \vezero$.
	Our goal now is to show that whenever there exists $k \in [\tau]$ with $|z_k - z_{k-1}| > 2\tau^3 + 1$ (where we define $z_0 := 0$ for convenience), then we can construct a non-zero integral vector $(\veg, \veh) \in \Ker_{\Z}(F)$ satisfying $(\veg, \veh) \sqsubseteq (\vex, \vez)$, which shows that $(\vex, \vez) \not\in \G(A)$.
	If no such index $k$ exists, it means that $\|\vex\|_\infty \leq \Oh(\tau^3)$ because $z_k - z_{k-1} = a_k x_k$ holds in $F$ (and~$\bar{F}^i)$.
	Moreover, if $|z_k - z_{k-1}| \leq 2\tau^3 + 1$ for all $k \in [\tau]$, then  $|z_k| \leq \Oh(\tau^4)$ for every $k \in [\tau]$, hence $\|(\vex, \vez)\|_\infty \leq \Oh(\tau^4)$.
	Note that, since the dimension of~$(\vex,\vez)$ is~$2\tau$, this also implies $\| (\vex,\vez) \|_1 \leq \Oh(\tau^5)$.
	Thus we now focus on the case when $\exists k \in [\tau]: |z_k - z_{k-1}| > 2\tau^3+1$.

	Let us now assume that $(z_k-z_{k-1})$ is positive and $z_k \geq \tau^3$.
	There are three other possible scenarios: when $(z_k-z_{k-1})$ is positive but $z_k < \tau^3$, or when $(z_k-z_{k-1})$ is negative and $z_k < -\tau^3$ or $z_k \geq -\tau^3$.
	We will later show that all these situations are symmetric to the one we consider and our arguments carry over easily, hence our assumption is without loss of generality.

	\lv{\begin{claim}}\sv{\begin{claim}[\appmark]}\label{clm:hillcutClaim}
		If $(z_k-z_{k-1})$ is positive and $z_k \geq \tau^3$, then there exists nonzero $(\veg,\veh) \in \Ker_{\Z}(F)$ with $(\veg,\veh) \sqsubseteq (\vex,\vez)$.
	\end{claim}
  \toappendix{
  	\lv{\begin{claimproof}}\sv{\begin{claimproof}[Proof of Claim~\ref{clm:hillcutClaim}]}
  		We distinguish two cases: either $z_k, \dots, z_\tau \geq \tau$, or there exists $\ell' > k$ such that $z_{\ell'} < \tau$.
  		In the first case, we set $g_k=1$ and $h_k = h_{k+1} = \cdots = h_\tau=a_k$, and all other coordinates of $(\veg, \veh)$ to zero.
  		It is straightforward to verify that $F \cdot (\veg, \veh) = \vezero$.
  		Clearly also $(\veg, \veh) \sqsubseteq (\vex, \vez)$ and we are done.

  		For the second case we assume there exists an index $\ell' > k$ such that $z_{\ell'} < \tau$.
  		We claim that then there exists an $\ell > k$ such that $x_\ell < -\tau$.
  		Suppose not: then we have $\sum_{q=k+1}^{\ell'} a_q x_q \geq \sum_{q=k+1}^{\ell'} \tau \cdot (-\tau) \geq -\tau^2$ and hence $z_{\ell'} = z_k + \sum_{q=k+1}^{\ell'} a_q x_q \geq \tau^3 - \tau^2 \geq \tau^2$, a contradiction (for~$\tau \geq 2$).
  		Let $\ell$ be the smallest index larger than $k$ satifying $x_\ell < -\tau$.
  		Note that by the minimality of~$\ell$ we have $z_q \geq \tau^2$ for all $q = k, k+1, \ldots, \ell-1$.
  		Now let $g_k := a_\ell$, $g_\ell := -a_k$, $h_k, h_{k+1}, \dots, h_{\ell-1} := a_k \cdot a_\ell$, and all remaining coordinates be zero.
  		It is easy to check that $(\veg, \veh) \in \Ker_{\Z}(F)$, and it remains to show that $(\veg, \veh) \sqsubseteq (\vex, \vez)$.
  		Since $a_kx_k = z_k - z_{k-1} \geq \tau^3$ we clearly have $x_k \geq \tau^2$ hence $g_k < x_k$, and since $x_\ell < -\tau$ we also have $|g_\ell| < |x_\ell|$, so $\veg \sqsubseteq \vex$.
  		Next, since $z_k \geq \tau^3$ and $\ell$ is smallest such that $x_\ell < -\tau$, the sequence $z_{k+1}, z_{k+2}, \dots, z_{\ell}$ decreases by at most $\tau^2$ in each step, and there are at most $\tau-1$ steps, hence $z_{q} \geq \tau^3 - (\tau-1)\tau^2 = \tau^2 \geq a_k \cdot a_\ell$ for all $q = k, \ldots, \ell-1$, concluding $\veh \sqsubseteq \vez$.
  	\end{claimproof}
  }
	Let us consider the remaining symmetric cases.
	If $z_k - z_{k-1}$ is negative and $z_k < -\tau^3$, then $(-\vex, -\vez)$ satisfies the original assumption, leading to some $(\veg', \veh') \sqsubseteq (-\vex, -\vez)$, hence $(-\veg', -\veh') \sqsubseteq (\vex, \vez)$ and we are done.
	If $z_k - z_{k-1}$ is negative but $z_k > -\tau^3$, then we would pick the largest index $\ell$ smaller than $k$ with $x_\ell > \tau$ and continue as before (the symmetry is that now $\ell$ is to the left of $k$ rather than to its right; that is, the case distinction from the previous paragraph is according to the value of $z_1$).
	Lastly, if $z_k - z_{k-1}$ is positive but $z_k < \tau^3$, negating $(\vex, \vez)$ gives a reduction to the previous case.
\end{proof}

\toappendix{
  \begin{remark}
  We have called Lemma~\ref{lem:hillcut} the ``Hill-cutting Lemma''.
  This is based on the intuitive visualization of the proof (a large coordinate corresponds to a hill, and subtracting a Graver element from it cuts the hill down, until no large hills are left).
  It is also a name known in the research of VASSes (vector addition systems with states) where a similar idea appears in a paper of Valiant and Patterson from 1975~\cite{ValiantP:1975}.
  \end{remark}
  \medskip
}
Together with the observation from the previous section and using our newly obtained bounds together with Corollary~\ref{cor:conflp_g1_tw}, we obtain:
\begin{corollary}\label{cor:solvingConfIPHMScheduling}
	Let~$I=(\ven, \vem, \vep, \vew)$ be an instance of $R|HM|C_{\max}$ or $R|HM|\sum w_jC_j$.
	A~\eqref{eq:conflp} optimum $\vey^*$ with $|\suppo(\vey^*)| \leq r + \bar{\tau}$ can be found in time $\poly(p_{\max}, \tau, \kappa, \la \ven, \vem, \vep, \vew \ra)$.
\end{corollary}

\section{Kernelization of Huge \texorpdfstring{$N$}{N}-fold IP}\label{sec:kernelizationOfHugeNFold}
\sv{\toappendix{\section{Additional Material for Section~\ref{sec:kernelizationOfHugeNFold}}}}

\subsection{Reducing by proximity} \label{ssec:reduceconflp}
Our first step in obtaining a kernel for quadratic huge $N$-fold IP is to solve~\eqref{eq:conflp}, as discussed previously (Lemma~\ref{lem:frac} and Corollary~\ref{cor:conflp_g1_tw}).
The next step is to obtain an equivalent (quadratic) huge $N$-fold IP instance whose components will be polynomially bounded, except for the objective.
For that, we derive the following proximity theorem.
We let $\{x\} = x - \floor{x}$ be the fractional part of $x$, and extend this definition coordinate-wise to vectors, i.e., $\{\vex\} = (\{x_1\}, \{x_2\}, \dots, \{x_n\})$.
\begin{reptheorem}{thm:hugeNFoldKernel}
  Let $\vey^*$ be an optimum of~\eqref{eq:conflp} with $|\suppo(\vey^*)| \leq r+\bar{\tau}$ and let
  \begin{equation*}
  P \df \left((r+\bar{\tau}) 26t^4 \log(t\|E^1_2,\dots,E^{\bar{\tau}}_2\|_\infty)\right)  (2r)^{r+1} (\|E\|_\infty s)^{3rs} \enspace .
  \end{equation*}
  Define $\vey^*_{-P} \df \max\{\vezero, \floor{\vey} - P \bm{1}\}$ (where the $\max$ is taken coordinate-wise).
  and for each $i \in [\bar{\tau}]$, let $\hat{\vecc}_i \df \frac{1}{\| \{\vey^i_{-P}\} \|_1} \sum_{\vecc \in \CC^i} \{y_{-P}(i, \vecc)\} \vecc$ and denote $\hat{\CC} = \{(i,\hat{\vecc}^i) \mid i \in [\bar{\tau}]\}$.
  There exists an optimal solution $\vezeta = \vey_{-P} + \bar{\vezeta}$ of~\eqref{eq:confilp} with $\|\bar{\vezeta}\|_1 \leq (2\bar{\tau} + r)P$ and
  \[
    \suppo(\bar{\vezeta}) \subseteq \{(i,\vecc) \mid i \in [\bar{\tau}], \,\vecc \in \CC^i, \, \exists (i, \vecc') \in \suppo(\vey) \cup \hat{\CC}: \|\vecc - \vecc'\|_\infty \leq P \} \enspace .
  \]
\end{reptheorem}
\begin{proof}
The proof is by a reinterpretation of a proximity theorem of Knop et al.~\cite[Theorem~21]{KnopKLMO:2019}.
Define a mapping $\varphi$ which assigns to each solution $\vey$ of~\eqref{eq:conflp} a solution $\vex$ of the relaxation of~\eqref{eq:hugenfold} as follows.
For each coordinate $(i,\vecc)$ such that $y(i,\vecc) > 0$, add to $\vex$ $\floor{y(i,\vecc)}$ bricks of type $i$ with configuration $\vecc$.
Then, for each $i \in [\bar{\tau}]$, add $\sum_{\vecc \in \CC^i} \{y(i, \vecc)\} = \|\{\vey^i\}\|_1 = \|\{\vey_{-P}^i\}\|_1$ bricks of type $i$ with value $\hat{\vecc}^i$, i.e., the remaining bricks take values which are the ``average of remaining configurations of type $i$.''
\begin{proposition}[{\cite[Theorem~21]{KnopKLMO:2019}}]
	\label{thm:old_proximity}
	Let $\vey$ be an optimum of~\eqref{eq:conflp} with $|\suppo(\vey)| \leq r+\bar{\tau}$ and $\vex^* = \varphi(\vey)$.
	Then there exists an optimal solution $\vez^*$ of~\eqref{eq:hugenfold} such that
	\begin{equation*}
	\|\vez^*-\vex^*\|_1 \leq \left((r+\bar{\tau}) 26t^4 \log(t\|E^1_2,\dots,E^\tau_2\|_\infty)\right)  (2r)^{r+1} (\|E\|_\infty s)^{3rs} = P \enspace .
	\end{equation*}
\end{proposition}
Proposition~\ref{thm:old_proximity} speaks about the~\eqref{eq:hugenfold} and its relaxation, and our task now is to interpret this as a statement about~\eqref{eq:confilp} and~\eqref{eq:conflp}.
There are two key observations: if $\vez^*$ and $\vex^*$ differ by at most $P$ in $\ell_1$-norm, then, first, they differ in at most $P$ bricks, and second, they differ in each brick by at most $P$ in $\ell_\infty$-norm.
This means that, first, all but at most $P$ bricks in $\vez^*$ take on values $\vecc'$ with $(i,\vecc') \in \suppo(\vey)$, and second, those which have some value $\vecc$ outside of $\suppo(\vey)$ are at $\ell_\infty$-distance at most $P$ from some $\vecc'$ with $(i,\vecc') \in \suppo(\vey) \cup \hat{\CC}$.
So, in a sense, the configurations $\vecc$ for $(i, \vecc) \in \suppo(\vey) \cup \hat{\CC}$ serve as centers around which the whole solution $\vezeta$ is concentrated.
Moreover, notice that by the fact that $\|\suppo(\vey)\| \leq r + \bar{\tau}$ and $|\hat{\CC}| = \bar{\tau}$, the number of these centers is at most $2\bar{\tau} + r$.

The construction of $\vey_{-P}$ reflects the first observation and, for each $(i, \vecc) \in \suppo(\vey) \cup \hat{\CC}$, corresponds to fixing the value to be $\vecc$ for $\max\{0,\floor{y(i,\vecc)-P}\}$ bricks.
Also the bound on $\|\bar{\vezeta}\|_1$ follows from the first observation: after we have fixed the value of all but $P$ bricks for each $(i, \vecc) \in \suppo(\vey) \cup \hat{\CC}$, there are at most $|\suppo(\vey) \cup \hat{\CC}| \cdot P = (2\bar{\tau} + r)P$ bricks whose value is to be decided.
The claim about $\suppo(\bar{\vezeta})$ follows from the second observation: each of the yet-to-be-assigned bricks has to lie near one of the ``centers''.
Since there is a 1:1 correspondence (up to brick permutation) between the solutions of~\eqref{eq:hugenfold} and~\eqref{eq:confilp}, the theorem follows.
\end{proof}
Theorem~\ref{thm:hugeNFoldKernel} intuitively says, that an optimum $\vey^*$ of~\eqref{eq:conflp} tells us around which at most $2\bar{\tau} + r$ ``centers'' will be an optimum of~\eqref{eq:confilp} concentrated.
Based on this, we now construct an equivalent but smaller instance of Huge $N$-fold.

Let $\vey^*$ be an optimum of~\eqref{eq:conflp} as above.
thLet $\tilde{\tau} = \bar{\tau} + |\suppo(\vey^*)|$; this is the number of ``centers''.
We shall now define an instance of huge $\bar{N}$-fold IP with $\tilde{\tau}$ types as follows.
For each $(i, \vecc') \in \suppo(\vey^*) \cup \hat{\CC}$, define a type with blocks $E^i_1$, $E^i_2$, with lower bound $\floor{\vecc'} - \max\{\floor{\vecc'} - P\veone, \vel^i\}$, upper bound $\floor{\vecc'} - \min\{\floor{\vecc'} + (P+1)\veone, \veu^i\}$, right hand side $\veb^i - E^i_2 \floor{\vecc'}$, and objective $\bar{f}^{i, \vecc'}(\bar{\vecc}) = f^i(\floor{\vecc'} + \bar{\vecc})$.
(Note that the new objective is separable quadratic if the old one was.)
Let $\bar{\mu}^{(i, \vecc')} = \min\{P, \floor{y(i, \vecc')}\}$ if $(i, \vecc') \in \suppo(\vey^*)$ and $\bar{\mu}^{(i, \vecc')} = \|\{(\vey^*)^i\}\|_1$ if $(i, \vecc') \in \hat{\CC}$.
Let $\bar{N} = \|\bar{\vemu}\|_1$.
Finally, let
\[
\bar{\veb}^0 \df \veb^0 - B\vey_{-P} -  \sum_{(i, \vecc') \in \suppo(\vey^*) \cup \hat{\CC}} \mu^{(i, \vecc')} E^i_1 \floor{\vecc'} \enspace .
\]
(In other words, we first define an instance whose bounds are restricted to $\CC^i \cap [\floor{\vecc'}-P\veone, \ceil{\vecc'}+P\veone]$ for each $(i, \vecc') \in \suppo(\vey^*) \cup \hat{\CC}$.
This means that the lower and upper bounds are close to each other but still potentially large
Then, we perform the variable transformation $\bar{\vecc} = \vecc - \floor{\vecc'}$ which ``shifts'' the origin of our coordinate system to $\vecc'$ for each brick of type $(i, \vecc')$, in order to shrink the bounds.)

The resulting instance is
\begin{equation} \label{eq:redhugeip}
\min \bar{f}(\bar{\vex}) \colon E^{(\bar{N})} \bar{\vex} = \bar{\veb}, \, \bar{\vel} \leq \bar{\vex} \leq \bar{\veu} \tag{reduced-HugeIP} \enspace .
\end{equation}
The equivalence of the original instance and~\eqref{eq:redhugeip} follows immediately from its definition and Theorem~\ref{thm:hugeNFoldKernel}.

\subsection{Objective Reduction} \label{ssec:reduceobj}
In~\eqref{eq:redhugeip}, everything except possibly $\bar{f}$ is bounded by a polynomial of the parameters.
The next lemma shows how to replace $\bar{f}$ with an equivalent but small function.
\begin{lemma}[Quadratic huge $N$-fold objective reduction] \label{lem:objreduction}
	Let a huge $N$-fold IP instance with a separable quadratic objective~$f$ be given explicitly by its coefficients, and denote $D \df \|\veu - \vel\|_\infty$.
	Then, there is a separable quadratic function $\tilde{f}$ which is equivalent to $f$ on $\{ \vex \mid \vel \le \vex \le \veu \}$ with coefficients bounded by $\coeff(f)$ such that both $\log \tilde{f}^{[\vel, \veu]}_{\max}, \log \coeff(\tilde{f}) \leq \poly(\bar{\tau}, t, \log (ND) )$.
	Moreover, $\tilde{f}$ can be computed in polynomial time.
\end{lemma}
\toappendix{
  \lv{\begin{proof}}\sv{\begin{proof}[Proof of Lemma~\ref{lem:objreduction}]}
  	Our aim is to use Proposition~\ref{thm:FT}.
  	Notice that the statement equivalently applies to any $d$-dimensional box of largest dimension $2M$ because the position of the box $[-M,M]^d$ relative to the origin does not matter.
  	There are two obstacles to straightforwardly applying this theorem to the assumed huge $N$-fold IP instance: first, the objective function is not a~linear function, and second, the dimension is too large, so we would get an encoding length polynomial in $N$ and not $\log N$.
  	Our approach to these obstacles is to first turn the objective into a linear one in an extended space (``linearization''), then use this linearity to aggregate variables across bricks of the same type (``aggregation''), then reduce the resulting linear objective using Proposition~\ref{thm:FT}, and then reverse the take steps, that is, ``deaggregate'' and ``delinearize'' the obtained function.

  	We first deal with the first obstacle.
  	By assumption, for each type $i \in [\bar{\tau}]$, each brick $j \in T_i$, and each coordinate $\ell \in [t]$, the contribution of the variable $x^j_\ell$ to the objective is $\alpha^{j}_\ell (x^j_\ell)^2 + \beta^{j}_\ell x^j_\ell + \gamma^{j}_\ell$ for some $\alpha^{j}_\ell, \beta^{j}_{\ell}, \gamma^{j}_\ell \in \Q$.
  	Since constant terms make no difference between solutions, from now on we assume $\gamma^{j}_\ell = 0$ for all $j$ and $\ell$.
  	We define an auxiliary variable $z^j_\ell = (x_\ell^j)^2$.\footnote{This can be thought of as a quadratic constraint.
  		One need not worry, though, as we are only constructing an equivalent objective to be used with the original instance, not actually introducing quadratic constraints into the formulation.
  		Finally, we are going to ``revert'' this step, since~$\tilde{f}$ is again a separable quadratic function.}
  	Now notice that the function becomes $\alpha^{j}_\ell z^{j}_\ell + \beta^{j}_\ell x^j_\ell$, hence linear over~$x^j_\ell$ and~$z^j_\ell$.
  	Having done this over all coordinates, we obtain an objective function $g(\vex, \vez) = (\vealpha, \vebeta)(\vez, \vex)$ which is linear over the variables $\vex, \vez$, and $\vex$ is a minimum for $f$ if and only if $\vex, \vez$ is a minimum for $g$ (we use the fact that $\vex$ uniquely determines $\vez$).
  	Moreover, assuming that $\vezero \in [\vel, \veu]$, it holds that if $\vex \in [\vel, \veu]$, then $\vez \in [(\vel)^2, (\veu)^2]$, where $(\vel)^2$ is obtained by squaring each component of $\vel$ and analogously for $\veu$.
  	In case $\vezero \not\in [\vel, \veu]$, we can obtain an equivalent instance which satisfies this condition by a simple affine transformation~\cite[Lemma 49]{EisenbrandEtAl2019}.

    \begin{figure}[bt]
      \input{tikz/objectiveReduction.tikz}
      \caption{\label{fig:objectiveReduction}%
        In the figure above, we have $P = \{ (\vex,(\vex)^2) \mid \vel\le\vex\le\veu \}$, $\vezero \in P$, $Q = [-ND^2, ND^2]^{2t\bar{\tau}}$.
        Each box corresponds to a stage in the process of reducing $f$ to $\tilde{f}$ and contains three elements: \begin{enumerate\sv{*}}[label=\bfseries(\arabic*)]
          \item a currently considered objective function equivalent to $f$, \item the space under consideration, and \item the dimension of this space.
        \end{enumerate\sv{*}}
        The change corresponding to each step is given a distinct color, which highlights the main feature which is being changed (e.g., changing the objective from quadratic to linear, shrinking the dimension, reducing the coefficients, etc.).
      }
    \end{figure}

  	Now we overcome the second obstacle.
  	For each $i \in [\bar{\tau}]$ and $\ell \in [t]$, we define $Z^{i}_\ell = \sum_{j \in T_i} z^{j}_\ell$ and $X^{i}_\ell = \sum_{j \in T_i} x^{j}_\ell$ and write $\veX, \veZ$.
  	Similarly define $A^{i}_\ell \df \mu_i \alpha^{j}_\ell$, $B^{i}_\ell \df \mu_i \beta^{j}_\ell$ for any $j \in T_i$, and denote $\veA \df (A^{1}_1, \dots, A^{\bar{\tau}}_t)$ and similarly $\veB$.
  	Note that, since we assume $\vezero \in [\vel,\veu]$ and $\| \vel - \veu \|_\infty = D$, we have that $\|\vex\|_\infty \le D$ and thus $\| \veX \|_\infty \le ND$ and $\|\veZ\|_\infty \le ND^2$, since also $\|\vez\|_\infty \le D^2$.
  	Consequently, we can assume that $(\veX,\veZ) \in Q = [-ND^2, ND^2]^{2 t \bar{\tau}}$.
  	Observe that minimizing $g(\vex, \vez)$ (and hence $f(\vex)$) is equivalent to minimizing $\sum_{i \in [\bar{\tau}]} \sum_{\ell \in [t]} A^i_\ell Y^i_\ell + B^i_\ell X^i_\ell = (\veA, \veB) (\veZ, \veX)$, which is a linear function over $2t\bar{\tau}$ variables which take on values at most $N D^2$ over the feasible region.
  	Applying the algorithm of Frank and Tardos to $(\veA, \veB)$ yields (in polynomial time) vector~$(\tilde{\veA}, \tilde{\veB})$ with $\|(\tilde{\veA}, \tilde{\veB})\|_\infty \leq (N D^2)^{\Oh(t \bar{\tau})^3}$ such that $(\tilde{\veA}, \tilde{\veB})$ is equivalent to $(\veA, \veB)$ on~$Q$.
  	Now we can work backwards to obtain an equivalent quadratic function $\tilde{f}$ in the original space of variables $\vex$.

  	\begin{claim}
  		Any function equivalent to $f$ on the variables $\vex$ on~$[-1,2]^{Nt}$ must have identical coefficients across bricks of the same type.
  	\end{claim}
  	\begin{claimproof}
  		Let $\tilde{f}$ be any separable quadratic function equivalent to $f$ on $[-1,2]^{Nt}$, and $\tilde{\alpha}^{j}_\ell, \tilde{\beta}^{j}_\ell$ be such that the contribution of variable $x^{j}_\ell$ of type $i$ is $\tilde{\alpha}^{j}_\ell (x^{j}_\ell)^2 + \tilde{\beta}^{j}_\ell x^{j}_\ell$.
  		In what follows we fix~$i \in [\tilde{\tau}]$.

  		We first prove that $\tilde{\beta}^{j}_\ell = \tilde{\beta}^{k}_\ell$ for each $j,k \in T_i$ and $\ell \in [t]$; fix $j,k,\ell$.
  		We define a point~$\vex$ which is~$0$ everywhere except two coordinates $x^j_\ell = 1$ and $x^k_\ell = -1$.
  		Now, we observe that $f(\vex) = f(-\vex)$.
  		Indeed, since $f$ and~$\tilde{f}$ are equivalent on $[-1,1]^{Nt} \subseteq [\vel,\veu]$, it must hold that $\tilde{f}(\vex) = \tilde{f}(-\vex)$.
  		This implies that
  		\[
  		\tilde{f}(\vex) = \tilde{\alpha}^{j}_\ell + \tilde{\beta}^{j}_\ell + \tilde{\alpha}^{k}_\ell - \tilde{\beta}^{k}_\ell
  		=
  		\tilde{\alpha}^{j}_\ell - \tilde{\beta}^{j}_\ell + \tilde{\alpha}^{k}_\ell + \tilde{\beta}^{k}_\ell = \tilde{f}(-\vex)\,.
  		\]
  		It follows that $2\tilde{\beta}^{j}_\ell = 2 \tilde{\beta}^{k}_\ell$ and we get $\tilde{\beta}^{j}_\ell = \tilde{\beta}^{k}_\ell$.

  		It remains to show that $\tilde{\alpha}^{j}_\ell = \tilde{\alpha}^{k}_\ell$ for each $j,k \in T_i$ and $\ell \in [t]$; again, fix $j,k,\ell$.
  		We define two points~$\vex$ and~$\hat{\vex}$ which are~$0$ everywhere except two coordinates: $x^j_\ell = 2, x^k_\ell = 1$ and $\hat{x}^j_\ell = 1, \hat{x}^k_\ell = 2$.
  		Observe that we have $f(\vex) = f(\hat{\vex})$ which implies~$\tilde{f}(\vex) = \tilde{f}(\hat{\vex})$.
  		Now, we get that
  		\[
  		\tilde{f}(\vex) = 4\tilde{\alpha}^{j}_\ell + 2\tilde{\beta}^{j}_\ell + \tilde{\alpha}^{k}_\ell + \tilde{\beta}^{k}_\ell
  		=
  		\tilde{\alpha}^{j}_\ell + \tilde{\beta}^{j}_\ell + 4\tilde{\alpha}^{k}_\ell + 2\tilde{\beta}^{k}_\ell = \tilde{f}(\hat{\vex})\,.
  		\]
  		Thus, since $\tilde{\beta}^{j}_\ell = \tilde{\beta}^{k}_\ell$, we have that $3\tilde{\alpha}^{j}_\ell = 3\tilde{\alpha}^{k}_\ell$ and the claim follows.
  	\end{claimproof}

  	We define, for each $i \in [\bar{\tau}]$, $j \in T_i$, and $\ell \in [t]$, $\tilde{\alpha}^{j}_\ell = \frac{\tilde{A}^i}{\mu_i}$, $\tilde{\beta}^{j}_\ell = \frac{\tilde{B}^i}{\mu_i}$, and define $\tilde{g}(\vex, \vez) = (\tilde{\vealpha}, \tilde{\vebeta}) (\vex, \vez)$.
  	Clearly, $\tilde{g}$ is equivalent to $g$ on $P$ by the equivalence of $(\veA, \veB)$ and $(\tilde{\veA}, \tilde{\veB})$ on~$Q$ and the claim above.
  	Finally, we define $\tilde{f}$ to be the separable quadratic function where the contribution of variable $x^j_\ell$, $j \in T_i$, is $\tilde{\alpha}^j_\ell (x^j_\ell)^2 + \tilde{\beta} x^j_\ell$.
  	By the definition of the variables $z^j_\ell$, a point $\vex \in [\vel, \veu]$ is a minimum for $\tilde{f}$ if and only if $(\vex, \vez)$ is a minimum for $\tilde{g}$.
  	Now, $(\vex,\vez)$ is a minimum for $\tilde{g}$ if and only if it is a minimum for $g$.
  	Thus, $\vex$ is a minimum for $f$, hence $\tilde{f}$ is equivalent with $f$ on $[\vel, \veu]$.
  	Furthermore, we get that $\log \tilde{f}_{\max}^{[\vel, \veu]}, \log \coeff(\tilde{f}) \leq \poly(\bar{\tau}, t, \log (ND) )$.
  \end{proof}
}
The consequence of Lemma~\ref{lem:frac}, Theorem~\ref{thm:hugeNFoldKernel}, and Lemma~\ref{lem:objreduction} is the following:
\begin{reptheorem}{thm:redhugeip}
	If~\eqref{eq:separation} is solvable in polynomial time, then quadratic huge $N$-fold IP admits a polynomial kernel when parameterized by $(\bar{\tau}, t, \|E\|_\infty)$.
\end{reptheorem}

\section{Conclusions and Research Directions}
We designed kernels for scheduling problems based on the~\eqref{eq:conflp}, which may be of practical interest due to the ubiquity of~\eqref{eq:conflp} (aka column generation) in practice.

On the side of theory, one may wonder why not apply the approach developed here to other scheduling problems, in particular those modeled as quadratic huge $N$-fold IP in~\cite{KnopKLMO:2019}, such as $R|r_j,d_j|\sum w_j C_j$.
The answer is simple: we are not aware of a way to solve the separation problem in polynomial time; in fact, we believe this to be a hard problem roughly corresponding to \textsc{Unary Vector packing} in variable dimension.
However, the typical use of Configuration LP is not to obtain an exact optimum (which is often hard), but to obtain an approximation which is good enough.
Perhaps a similar approach within our context may lead to so-called \emph{lossy kernels}~\cite{lokshtanov2017lossy}?
However, it is not even clear that an approximate analogue of Theorem~\ref{thm:hugeNFoldKernel} holds, because getting an LP solution whose value is close to optimal does not immediately imply getting a solution which is (geometrically) close to some optimum; cf. the discussions on $\epsilon$-accuracy in~\cite[Definition 31]{EisenbrandEtAl2019} and~\cite[Introduction]{HochbaumShantikumar1990}.



\bibliography{scheduling}

\appendix
\clearpage

\appendixText

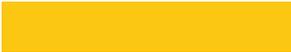
\end{document}

%% file: macros/complexity.tex
\newcommand{\Oh}{\mathcal{O}}
\newcommand{\OhOp}[1]{\Oh\mathopen{}\mathclose\bgroup\left( #1 \aftergroup\egroup\right)}
\DeclareMathOperator{\poly}{{\rm poly}}

\usepackage{xspace}
\newcommand{\FPT}{{\sf FPT}\xspace}
\newcommand{\NP}{{\sf NP}\xspace}

\usepackage{tabularx}

%% file: macros/todo.tex
\usepackage[textwidth=20mm,obeyFinal,colorinlistoftodos]{todonotes}
\newcommand{\mytodo}[2]{\todo[size=\tiny, color=#1!50!white]{#2}\xspace}

\newcommand{\mkcom}[1]{\mytodo{blue}{#1}}

%% file: macros/mics.tex
\newcommand{\hy}{\hbox{-}\nobreak\hskip0pt}

\newcommand{\suppo}{\mathrm{supp}}

\newcommand{\tw}{\textrm{tw}}

\newcommand{\coeff}{\mathrm{coeff}}
\newcommand{\sep}{\mathfrak{sep}}

\newcommand{\CC}{\mathcal{C}}

\DeclarePairedDelimiter\ceil{\lceil}{\rceil}
\DeclarePairedDelimiter\floor{\lfloor}{\rfloor}

\usepackage{stackengine}
\stackMath

\newcommand{\intcone}{\mathrm{cone}_{\N}}
\newcommand{\df}{:=}

%% file: macros/nfold.tex
\def\ve#1{\mathchoice{\mbox{\boldmath$\displaystyle\bf#1$}}
{\mbox{\boldmath$\textstyle\bf#1$}}
{\mbox{\boldmath$\scriptstyle\bf#1$}}
{\mbox{\boldmath$\scriptscriptstyle\bf#1$}}}
\newcommand\vea{{\ve a}}
\newcommand\veb{{\ve b}}
\newcommand\vecc{{\ve c}}
\newcommand\ved{{\ve d}}

\newcommand\veg{{\ve g}}
\newcommand\veh{{\ve h}}

\newcommand\vel{{\ve l}}
\newcommand\ven{{\ve n}}
\newcommand\vem{{\ve m}}

\newcommand\vep{{\ve p}}

\newcommand\veu{{\ve u}}
\newcommand\vev{{\ve v}}
\newcommand\vew{{\ve w}}
\newcommand\vex{{\ve x}}
\newcommand\vey{{\ve y}}
\newcommand\vez{{\ve z}}
\newcommand\vealpha{{\ve \alpha}}
\newcommand\vebeta{{\ve \beta}}

\newcommand\velambda{{\ve \lambda}}
\newcommand\vezeta{{\ve \zeta}}
\newcommand\vemu{{\ve \mu}}
\newcommand\vezero{{\ve 0}}
\newcommand\veone{{\ve 1}}
\newcommand\veA{{\ve A}}
\newcommand\veB{{\ve B}}

\newcommand\veX{{\ve X}}

\newcommand\veZ{{\ve Z}}
\def\R{\mathbb{R}}
\def\Z{\mathbb{Z}}
\def\N{\mathbb{N}}
\def\Q{\mathbb{Q}}
\def\G{\mathcal{G}}
\def\Ker{\textrm{Ker}}
\def \la {\langle}
\def \ra {\rangle}

%% file: macros/appHandler.tex
\usepackage{etoolbox}

\newcommand{\sv}[1]{}
\newcommand{\lv}[1]{#1}

\newcommand{\appendixText}{}

\newcommand{\toappendix}[1]{#1}

\newcommand{\appmark}{$\star$}

%% file: tikz/objectiveReduction.tikz
\begin{tikzpicture}[node distance=.5cm]
\tikzstyle{fitter}=[draw,inner sep=1pt]
\newcommand{\QuadToLinColor}{orange!70}
\newcommand{\dimRedColor}{violet!50}
\newcommand{\objRedColor}{yellow!60!black}
\newcommand{\dimExpansionColor}{red!50}
\newcommand{\LinToQuadColor}{blue!50}

\begin{scope}
  \node[fill=\QuadToLinColor,minimum width=1.6cm] (qh-obj) {$f(\vex)$};
  \node[below of=qh-obj] (qh-poly) {$\vex \in [\vel,\veu]$};
  \node[below of=qh-poly] (qh-dim) {$Nt$};

  \node[fitter,fit=(qh-obj)(qh-poly)(qh-dim)] (qhfit) {};
\end{scope}

\begin{scope}[yshift=-3cm]
  \node[fill=\LinToQuadColor] (qh-obj) {$\tilde{f}(\vex)$};
  \node[below of=qh-obj] (qh-poly) {$\vex \in [\vel,\veu]$};
  \node[below of=qh-poly] (qh-dim) {$Nt$};

  \node[below of=qh-dim,yshift=-.2cm,xshift=1cm] (qh-ex1) {$\tilde{f}_{\max} = Nt (ND^2)^{O((t\bar{\tau})^3)} \vex_{\max}$};
  \node[below of=qh-ex1] (qh-ex2) {$\langle \coeff(\tilde{f}) \rangle = O((t\bar{\tau})^3) \log(ND) + \|\vemu\|_\infty$};

  \node[fitter,fit=(qh-obj)(qh-poly)(qh-dim)] (qhrfit) {};
\end{scope}

\begin{scope}[xshift=5.5cm]
  \node[fill=\QuadToLinColor] (lh-obj) {$g(\vex,\vez) = (\vealpha,\vebeta)(\vez,\vex)$};
  \node[below of=lh-obj] (lh-poly) {$(\vex,\vez) \in P$};
  \node[below of=lh-poly,fill=\dimRedColor] (lh-dim) {$2Nt$};

  \node[fitter,fit=(lh-obj)(lh-poly)(lh-dim)] (lhfit) {};
\end{scope}

\begin{scope}[xshift=5.5cm,yshift=-3cm]
  \node[fill=\LinToQuadColor] (lh-obj) {$\tilde{g}(\vex,\vez) = (\tilde{\vealpha},\tilde{\vebeta})(\vez,\vex)$};
  \node[below of=lh-obj] (lh-poly) {$(\vex,\vez) \in P$};
  \node[below of=lh-poly,fill=\dimExpansionColor] (lh-dim) {$2Nt$};
  \node[below of=lh-dim, xshift=1.5cm,yshift=-.2cm] (lh-ex1) {$\langle \tilde{\alpha}_i \rangle = O((t\bar{\tau})^3) \log(ND) + \|\vemu\|_\infty$};
  \node[below of=lh-ex1] (lh-ex2) {$\langle \tilde{\beta}_i \rangle = O((t\bar{\tau})^3) \log(ND) + \|\vemu\|_\infty$};

  \node[fitter,fit=(lh-obj)(lh-poly)(lh-dim)] (lhrfit) {};
\end{scope}

\begin{scope}[xshift=10cm]
  \node[fill=\objRedColor] (lrd-obj) {$(\veA,\veB)(\veX,\veZ)$};
  \node[below of=lrd-obj] (lrd-poly) {$(\veX,\veZ) \in Q$};
  \node[below of=lrd-poly,fill=\dimRedColor] (lrd-dim) {$2t\bar{\tau}$};

  \node[fitter,fit=(lrd-obj)(lrd-poly)(lrd-dim)] (lrdfit) {};
\end{scope}

\begin{scope}[xshift=10cm,yshift=-3cm]
  \node[fill=\objRedColor] (lrd-obj) {$(\tilde{\veA},\tilde{\veB})(\veX,\veZ)$};
  \node[below of=lrd-obj] (lrd-poly) {$(\veX,\veZ) \in Q$};
  \node[below of=lrd-poly,fill=\dimExpansionColor] (lrd-dim) {$2t\bar{\tau}$};

  \node[fitter,fit=(lrd-obj)(lrd-poly)(lrd-dim)] (rlrdfit) {};
\end{scope}

  \draw[->,line width=4pt,color=\QuadToLinColor] (qhfit) to node[midway,yshift=.8cm,color=black] {\rotatebox{45}{linearize}} (lhfit);
  \draw[->,line width=4pt,color=\dimRedColor] (lhfit) to node[midway,yshift=.8cm,color=black] {\rotatebox{45}{aggregate}}(lrdfit);
  \draw[->,line width=4pt,color=\objRedColor] (lrdfit) to node[midway,xshift=.7cm,color=black] {reduce} (rlrdfit);
  \draw[->,line width=4pt,color=\dimExpansionColor] (rlrdfit) to node[midway,yshift=0.9cm,xshift=.3cm,color=black] {\rotatebox{45}{deaggregate}} (lhrfit);
  \draw[->,line width=4pt,color=\LinToQuadColor] (lhrfit) to node[midway,yshift=0.9cm,color=black] {\rotatebox{45}{delinearize}} (qhrfit);
\end{tikzpicture}